\documentclass[letterpaper, 10 pt, conference]{ieeeconf}%
\usepackage{bm}
\usepackage{amssymb}
\usepackage{epsfig}
\usepackage{subfigure}
\usepackage{amsmath}
\usepackage{graphicx}
\usepackage{epstopdf}
\usepackage{amsfonts}%
\usepackage{indentfirst}
\usepackage{eufrak}%
\usepackage[ruled]{algorithm2e}
\usepackage{color}

\newtheorem{problem}{\textbf{Problem}}
\newtheorem{definition}{\rm\textbf{Definition}}
\newtheorem{theorem}{\rm\textbf{Theorem}}
\newtheorem{lemma}{\rm\textbf{Lemma}}

\newtheorem{assump}{\rm\textbf{Assumption}}
\newtheorem{remark}{\rm\textbf{Remark}}
\providecommand{\U}[1]{\protect\rule{.1in}{.1in}}
\IEEEoverridecommandlockouts
\begin{document}

\title{{\LARGE \textbf{High Order Control Lyapunov-Barrier Functions \\ for Temporal Logic Specifications}}}
\author{Wei Xiao, Calin A. Belta and Christos G. Cassandras \thanks{This work was supported in part by NSF under grants IIS-1723995, CPS-
1446151, ECCS-1509084, DMS-1664644, CNS-1645681, by AFOSR under
grant FA9550-19-1-0158, by ARPA-E’s NEXTCAR program under grant
DE-AR0000796 and by the MathWorks.}\thanks{The authors are
with the Division of Systems Engineering and Center for Information and
Systems Engineering, Boston University, Brookline, MA, 02446, USA
\texttt{{\small \{xiaowei\}@bu.edu}}}}
\maketitle

\begin{abstract}
Recent work has shown that stabilizing an affine control system to a desired state while optimizing a quadratic cost subject to state and control constraints can be reduced to a sequence of Quadratic Programs (QPs) by using Control Barrier Functions (CBFs) and Control Lyapunov Functions (CLFs). In our own recent work, we defined High Order CBFs (HOCBFs) for systems and  constraints with arbitrary relative degrees. In this paper, in order to accommodate initial states that do not satisfy the state constraints and constraints with arbitrary relative degree, we generalize HOCBFs to High Order Control Lyapunov-Barrier Functions (HOCLBFs). 
We also show that the proposed HOCLBFs 
can be used to guarantee the Boolean satisfaction of Signal Temporal Logic (STL) formulae over the state of the system. We illustrate our approach on a safety-critical optimal control problem (OCP) for a unicycle. \end{abstract}

\thispagestyle{empty} \pagestyle{empty}


\section{Introduction}
\label{sec:intro}

Barrier functions (BFs) are
Lyapunov-like functions \cite{Tee2009}, whose use can be traced back to optimization problems \cite{Boyd2004}. More recently, they have been employed to prove set invariance \cite{Aubin2009}, \cite{Prajna2007} and for the purpose of multi-objective control \cite{Panagou2013}. In \cite{Tee2009}, it was proved that if a BF for a given set satisfies Lyapunov-like conditions, then the set is forward invariant. A less restrictive form of a BF, which is allowed to grow when far away from the boundary of the set, was proposed in \cite{Ames2017}.
Another approach that allows a BF to take zero values was proposed in \cite{Glotfelter2017}, \cite{Lindemann2018}. 
Control BFs (CBFs) are extensions of BFs for control systems, and are used to map a constraint that is defined over system states to a constraint on the control input. Recently, it has been shown that, to stabilize an affine control system while optimizing a quadratic cost and satisfying state and control constraints, CBFs can be combined with
control Lyapunov functions (CLFs) \cite{Artstein1983}, \cite{Freeman1996}, \cite{Aaron2012} to form quadratic programs (QPs) \cite{Ames2017}, \cite{Glotfelter2017} that are solved in real time.

The CBFs from \cite{Ames2017} and \cite{Glotfelter2017} work for constraints that have relative degree one with respect to the system dynamics. A CBF method for position-based constraints with relative degree two was proposed in \cite{Wu2015}. A more general form, which works
for arbitrarily high relative degree constraints, was proposed in \cite{Nguyen2016}. The method in \cite{Nguyen2016} employs input-output linearization and finds a pole placement controller with negative poles to stabilize the CBF to zero. In our recent work \cite{Xiao2019}, we defined a High Order CBF (HOCBF) that can accommodate state constraints with high relative degree and does not require linearization. In this paper, we propose an extension of the HOCBF from 
\cite{Xiao2019} that achieves two main objectives: (1) it works for states that are not initially in the safe set, and (2) it can guarantee the satisfaction of specifications given as Signal Temporal Logic (STL) formulae. 


Recent works proposed the use of CBFs to enforce the satisfaction of temporal logic (TL) specifications. STL and Linear TL (LTL) were used as specification languages in \cite{Lindemann2018} and \cite{Nillson2018}, respectively, for systems and constraints with relative degree one. This is restrictive to rich specifications for high relative degree systems. For example, we have comfort (jerk) requirement on autonomous vehicles, and this can only be achieved on high relative degree systems. The authors of \cite{Lindemann2018} defined time-varying functions to guarantee the satisfaction of a STL formula for systems with relative degree one. Extending time-varying functions to work for high relative degree constraints, even though possible, would be difficult, as it would require that the state of the system be in the intersection of a possibly large number of sets. TL specifications have also been considered in \cite{Srinivasan2018} by using finite-time convergence CBFs \cite{Li2018}. However, this approach is restricted to relative-degree-one constraints, and may lead to chattering behaviors that result from finite-time convergence, as will be shown in this paper.  Barrier-Lyapunov functions, as proposed in  \cite{Tee2009},  \cite{Sachan2018}, could also be used, in principle, to implement STL specifications, as they combine (linear) state constraints with convergence. 

In this paper, to accommodate STL specifications over nonlinear state constraints for high relative degree systems, we propose High Order Control Lyapunov-Barrier Functions (HOCLBF). The proposed HOCLBFs lead to controllers that stabilize a system inside a set within a specified time if the system state is initially outside this set, and ensure that the system remains in this set after it enters it.  We also propose how to eliminate chattering behaviors with the HOCLBF method.
We illustrate the usefulness of the proposed approach by applying it to a unicycle model.


\section{Preliminaries}
\label{sec:pre}

 \begin{definition} \label{def:classk}
 	({\it Class $\mathcal{K}$ and Extended Class $\mathcal{K}$ functions}) \cite{Khalil2002}) A continuous function $\alpha:\mathbb{R}\rightarrow\mathbb{R}$ is an extended class $\mathcal{K}$ function if it is strictly increasing and $\alpha(0)=0$. If $\alpha:[0,a)\rightarrow[0,\infty), a > 0$, then $\alpha$ belongs to class $\mathcal{K}$. 
 \end{definition}

 Consider an affine control system of the form
 \begin{equation} \label{eqn:affine}
 \dot {\bm{x}} = f(\bm x) + g(\bm x)\bm u
 \end{equation}
 where  $\bm x\in \mathbb{R}^n$, $f:\mathbb{R}^n\rightarrow \mathbb{R}^n$ and $g:\mathbb{R}^n \rightarrow \mathbb{R}^{n\times q}$ are globally Lipschitz, and $\bm u\in U \subset \mathbb{R}^q$ ($U$ denotes the control constraint set). Solutions $\bm x(t)$ of (\ref{eqn:affine}), starting at $\bm x(0)$, $t\geq 0$, are forward complete for all $\bm u\in U$.
 
 Suppose the control bound $U$ is defined as (the inequality is interpreted componentwise, $\bm u_{min},\bm u_{max}\in \mathbb{R}^q$):
 \begin{equation} \label{eqn:control}
 U:= \{\bm u\in\mathbb{R}^q: \bm u_{min}\leq\bm u \leq \bm u_{max}\}.
 \end{equation}
 
 \subsection{High Order Control Barrier Functions}
 \label{sec:hocbf}
 \begin{definition} \label{def:forwardinv}
 	A set $C\subset\mathbb{R}^n$ is forward invariant for system (\ref{eqn:affine}) if its solutions starting at any $\bm x(0) \in C$ satisfy $\bm x(t)\in C$ for $\forall t\geq 0$.
 \end{definition} 
 
 \begin{definition} \label{def:relative}
 	({\it Relative degree})
 	The relative degree of a (sufficiently many times) differentiable function $b:\mathbb{R}^n\rightarrow \mathbb{R}$ with respect to system (\ref{eqn:affine}) is the number of times we need to differentiate it along its dynamics until $\bm u$ explicitly shows in the corresponding derivative for some $\bm x$. 
 \end{definition}
 
 In this paper, since function $b$ is used to define a constraint $b(\bm x)\geq 0$, we will also refer to the relative degree of $b$ as the relative degree of the constraint. 
 
 For a constraint $b(\bm x)\geq 0$ with relative degree $m$, $b: \mathbb{R}^n \rightarrow \mathbb{R}$, and $\psi_0(\bm x) := b(\bm x)$, we define a sequence of functions  $\psi_i: \mathbb{R}^n \rightarrow \mathbb{R}, i\in\{1,\dots,m\}$:
 \begin{equation} \label{eqn:functions}
 \begin{aligned}
 \psi_i(\bm x) := \dot \psi_{i-1}(\bm x) + \alpha_i(\psi_{i-1}(\bm x)),\;\;\;i\in\{1,\dots,m\},
 \end{aligned}
 \end{equation}
 where $\alpha_i(\cdot),i\in\{1,\dots,m\}$ denotes a $(m-i)^{th}$ order differentiable class $\mathcal{K}$ function.
 
 We further define a sequence of sets $C_i, i\in\{1,\dots,m\}$ associated with (\ref{eqn:functions}) in the form:
 \begin{equation} \label{eqn:sets}
 \begin{aligned}
 C_i := \{\bm x \in \mathbb{R}^n: \psi_{i-1}(\bm x) \geq 0\}, \;\;\;i\in\{1,\dots,m\}.
 \end{aligned}
 \end{equation}

 \begin{definition} \label{def:hocbf}
 	({\it High Order Control Barrier Function (HOCBF)} \cite{Xiao2019}) Let $C_1, \dots, C_{m}$ be defined by (\ref{eqn:sets}) and $\psi_1(\bm x), \dots, \psi_{m}(\bm x)$ be defined by (\ref{eqn:functions}). A function $b: \mathbb{R}^n\rightarrow \mathbb{R}$ is a high order control barrier function (HOCBF) of relative degree $m$ for system (\ref{eqn:affine}) if there exist $(m-i)^{th}$ order differentiable class $\mathcal{K}$ functions $\alpha_i,i\in\{1,\dots,m-1\}$  and a class $\mathcal{K}$ function $\alpha_{m}$ such that
 	{\small\begin{equation}\label{eqn:constraint}
 	\begin{aligned}
 	\sup_{\bm u\in U}[L_f^{m}b(\bm x) \!+\! L_gL_f^{m-1}b(\bm x)\bm u \!+\! S(b(\bm x)) \!+\! \alpha_m(\psi_{m-1}(\bm x))] \geq 0,
 	\end{aligned}
 	\end{equation}
 	}for all $\bm x\in C_1 \cap,\dots, \cap C_{m}$. 
 \end{definition}
  
  In (\ref{eqn:constraint}), $L_f^m$ ($L_g$) denotes Lie derivatives along $f$ ($g$) $m$ (one) times, $S(\cdot)$ denotes the remaining Lie derivatives along $f$ with degree $<m$ (omitted for simplicity, see \cite{Xiao2019}). Assume the number of $\bm x$ such that $L_gL_f^{m-1}b(\bm x) = 0$ is finite.
 
 \begin{theorem} \label{thm:hocbf}
 	(\cite{Xiao2019}) Given a HOCBF $b(\bm x)$ from Def. \ref{def:hocbf} with the associated sets $C_1, \dots, C_{m}$ defined by (\ref{eqn:sets}), if $\bm x(0) \in C_1 \cap,\dots,\cap C_{m}$, then any Lipschitz continuous controller $\bm u(t)$ that satisfies (\ref{eqn:constraint}), $\forall t\geq 0$, renders
 	$C_1\cap,\dots, \cap C_{m}$ forward invariant for system (\ref{eqn:affine}).
 \end{theorem}

The HOCBF is a general form of the relative degree one CBF \cite{Ames2017}, \cite{Glotfelter2017}, \cite{Lindemann2018} (setting $m = 1$ reduces the HOCBF to the common CBF form in \cite{Ames2017}, \cite{Glotfelter2017}, \cite{Lindemann2018}). In order to accomodate initial conditions $\bm x(0)$ that are not in $C_1$, the extended class $\mathcal{K}$ functions are used in the definition of a relative degree one CBF \cite{Xu2015}, \cite{Ames2017}. In this way, a system will be assymptotically stabilized to a safe set that is defined by a safety constraint if the system is initially outside this set, but this may not work for high relative degree constraints, as will be shown in the next section. The HOCBF is also a general form of the exponential CBF \cite{Nguyen2016}. 

For system (\ref{eqn:affine}), consider the following cost:
\begin{equation}\label{eqn:cost}
J(\bm u(t)) = \int_{0}^{T}\mathcal{C}(||\bm u(t)||)dt
\end{equation}
where $||\cdot||$ denotes the 2-norm of a vector, and  $\mathcal{C}(\cdot)$ is a strictly increasing function.

\begin{problem}[Optimal Control Problem (OCP)] \label{prob:tocp}
Given system (\ref{eqn:affine}) with initial condition $\bm x(0)$,
find a control law that minimizes cost (\ref{eqn:cost}), while satisfying the control bounds (\ref{eqn:control}) and a constraint $b(\bm x)\geq 0$, for all $t\in [0,T]$.
\end{problem}  

Under the assumption that the cost
(\ref{eqn:cost}) is quadratic, the above OCP can be (conservatively) reduced to a sequence of quadratic programs (QP),
by discretizing the time, keeping the state constant at its value at the beginning of each interval, and solving for a constant optimal control in each interval (note that constraint (\ref{eqn:constraint}) is linear in control when the state is constant). Most existing approaches use a simpler form of (\ref{eqn:constraint}), which corresponds to a constraint of relative degree 1 \cite{Ames2017}, \cite{Lindemann2018}, \cite{Nguyen2016}. HOCBFs are used for arbitrary relative degree constraints in \cite{Xiao2019}. To guarantee the QP feasibility, we can use the analytical approach \cite{Xiao2021}, learning methods \cite{Xiao2020CDC}  or adaptive CBF methods \cite{Xiao2020}.

\subsection{Signal Temporal Logic (STL)}
\label{sec:stl}
In this paper, we use the negation-free signal temporal logic (STL) \cite{Maler2004} to specify regions of interest to be reached by the states of system (\ref{eqn:affine}). Its syntax is given by
\begin{equation} \label{eqn:task}
\varphi := \mu\vert\varphi_1\wedge \varphi_2\;\vert\; \varphi_1\vee \varphi_2\vert \varphi_1\!\Rightarrow\! \varphi_2 \vert\mathcal{G}_{I}\varphi \;\vert\; \mathcal{F}_{I}\varphi\; \vert\; \varphi_1\mathcal{U}_{I}\varphi_2
\end{equation}
where $\mu :=(b(\bm x)\geq 0)$ is a predicate on the state vector $\bm x$ of (\ref{eqn:affine}) and $b:\mathbb{R}^n\rightarrow \mathbb{R}$ is a differentiable function of relative degree $m$ with respect to system (\ref{eqn:affine}). $\wedge, \vee$ are Boolean operators for conjunction and disjunction, respectively; $\Rightarrow$ is the Boolean implication, and $\mathcal{U}_{I},  \mathcal{F}_{I}, \mathcal{G}_{I}$ are the timed ``until'', ``eventually'', and ``always'' operators, respectively, and $I=[t_a,t_b]$ is a time interval, $t_b\geq t_a\geq 0$. Note that the absence of negation does not restrict the expressivity of STL \cite{Quak2008}. Let $(\bm x, t)$ denote the trajectory $\bm x$ of (\ref{eqn:affine}) starting at time $t$. The (Boolean) semantics of STL are inductively defined as
\begin{equation} \label{eqn:semantics}
\begin{aligned}
(\bm x,t) \models \mu   &\Leftrightarrow  b(\bm x(t))\geq 0\\
(\bm x,t) \models \varphi_1 \wedge \varphi_2   &\Leftrightarrow (\bm x,t) \models \varphi_1\wedge (\bm x,t) \models \varphi_2\\
(\bm x,t) \models \varphi_1 \vee \varphi_2   &\Leftrightarrow (\bm x,t) \models \varphi_1\vee (\bm x,t) \models \varphi_2\\
(\bm x,t) \models \varphi_1 \mathcal{U}_{I}\varphi_2  &\Leftrightarrow \exists t^{'} \in t + I \text{ s.t. } (\bm x,t^{'})\models \varphi_2 \\ &\;\;\;\;\;\wedge \forall t^{''}\in[t_a,t^{'}], (\bm x,t^{''})\models \varphi_1 \\
(\bm x,t) \models \mathcal{F}_{I}\varphi  &\Leftrightarrow\exists t^{'}\in t+I \text{ s.t. } (\bm x,t^{'})\models\varphi \\
(\bm x,t) \models \mathcal{G}_{I}\varphi  &\Leftrightarrow\forall t^{'}\in t+I, (\bm x,t^{'})\models\varphi\\
(\bm x,t) \models (\varphi_1 \Rightarrow\varphi_2)  &\Leftrightarrow(\bm x,t)\! \models\! \varphi_1 \text{ implies } (\bm x,t)\! \models\! \varphi_2
\end{aligned}
\end{equation}
where $\models$ denotes the satisfaction relation, and $t+I$ denotes $[t+t_a,t+t_b]$. We say that a trajectory $\bm x$ satisfies a formula $\varphi$ if it is satisfied at time 0, i.e., $(\bm x, 0)\models \varphi$. For simplicity, we will write $\bm x\models \varphi$ to denote that $\bm x$ satisfies $\varphi$.

\textbf{Example:} Consider a unicycle model:
\begin{equation}\label{eqn:robot}
\left[\begin{array}{c} 
\dot x\\
\dot y\\
\dot \theta
\end{array} \right]=
\left[\begin{array}{c}  
v\cos(\theta)\\
v\sin(\theta)\\
0 \\
\end{array} \right]  + 
\left[\begin{array}{c}  
0\\
0\\
1
\end{array} \right] 
u
\end{equation}
where $(x,y)$ denote the coordinates of the robot, $v > 0$ denotes its linear speed, $\theta$ is its heading angle, and $u$ denotes its control (angular speed). 

Formula $\varphi_1:= \mathcal{G}_{[5,6]}(x^2(t) + y^2(t) \leq R^2)$, $R > 0$, requires the robot to satisfy the constraint $x^2(t) + y^2(t) \leq R^2$ for all times in $[5s, 6s] $. Formula $\varphi_2:= \mathcal{F}_{[5,6]}(x^2(t) + y^2(t) \leq R^2)$, $R  > 0$, requires the robot to satisfy the constraint $x^2(t) + y^2(t) \leq R^2$ for at least a time instant in $[5s,6s]$.

\section{Problem Formulation and Approach}
\label{sec:prob}
 In this paper, we consider the following problem:


\begin{problem}[OCP with STL constraints]\label{prob:pre}
Given system (\ref{eqn:affine}) with initial state $\bm x(0)$, and given a STL formula $\varphi$ over its state $\bm x$, 
find a control law that minimizes cost (\ref{eqn:cost}), while satisfying the control bounds (\ref{eqn:control}) and formula $\varphi$. 
\end{problem}

Assume the STL formula $\varphi$ can be satisfied for some controllers. In the case that it cannot be satisfied, we explore to maximally satisfy it, i.e., to maximize the STL robustness. This will be further studied in future work.

Our approach to Problem \ref{prob:pre} is based on two types of 
HOCLBF (class 1 and class 2, shown in the next section) and it can be summarized as follows. First, by exploiting the negation-free structure of formula $\varphi$, we break it down (assume it is tractable) into a set of atomic formulae
of the type $\mathcal{G}_{[t_a,t_b]}(b(\bm x(t))\geq 0)$ and $\mathcal{F}_{[t_a,t_b]}(b(\bm x(t))\geq 0)$. Starting from time $t=0$, we use a receding horizon $H > 0$ to determine the atomic formulae that we will consider at time $t$, i.e., we only consider the atomic formulae such that $[t,t+H]\cap[t_a,t_b]\ne \emptyset$ (the choice of $H$ is discussed 
at the end of Sec. \ref{sec:gcbf}).
For each predicate involved in these formulae, we define a HOCLBF (we discuss  later how to address possible conflicts among these predicates). If the current state satisfies a predicate $b(\bm x(t))\geq 0$ (the predicate most likely corresponds to a safety requirement), then we use a class 2 HOCLBF, which is a HOCBF as defined in our previous work \cite{Xiao2019}, to derive a controller that makes sure the predicate stays true for all future times.  If the current state does not satisfy the predicate (usually related to a state convergence requirement), we use a class 1 HOCLBF that makes sure the system satisfies the predicate before $t_b$ for atomic formulae with $\mathcal{F}_{[t_a,t_b]}$, and before $t_a$ for atomic formulae with $\mathcal{G}_{[t_a,t_b]}$. Once the predicate is satisfied, we switch to a class 2 HOCLBF. We show how the satisfaction of general STL formulae can be enforced with such class 1 and class 2 HOCLBFs.


\section{High Order Control Lyapunov-Barrier Functions}
\label{sec:gcbf}

In this section, we define high order control Lyapunov-barrier functions (HOCLBFs) for system (\ref{eqn:affine}), and classify them into two classes to accommodate systems with arbitrary initial states.

\textbf{Example revisited:} Consider the robot from the previous example and formula $\varphi_1$, which requires the satisfaction of constraint $x^2(t) + y^2(t) \leq R^2$ for all times in $[5s,6s]$.
This constraint has relative degree 2 for system (\ref{eqn:robot}). If this constraint is satisfied at time 0, then we can define a HOCBF $b(\bm x):= R^2 - x^2(t) - y^2(t)$ such that $\varphi_1$ is guaranteed to be satisfied if a controller $u$ satisfies the corresponding HOCBF constraint (\ref{eqn:constraint}). Otherwise, we cannot define a HOCBF for it since $b(\bm x(0)) < 0$ and the class $\mathcal{K}$ function $\alpha_1(\cdot)$ in (\ref{eqn:functions}) only allows for a non-negative argument. Thus, it is impossible to construct the corresponding sets $C_1, C_2$.

 If $b(\bm x(0)) < 0$ and $\dot b(\bm x(0)) > 0$, we can then redefine $\psi_i(\bm x)$ ($ i\in\{1,2\}$ in this case) in (\ref{eqn:functions}) as:
 \begin{equation}\label{eqn:example}
\begin{aligned}
\psi_1(\bm x):= \dot\psi_0(\bm x) + p_1\beta_1(\psi_0(\bm x)),\\
\psi_2(\bm x):= \dot\psi_1(\bm x) + \alpha_2(\psi_1(\bm x)),\end{aligned}
\end{equation}
 where $\psi_0(\bm x)= b(\bm x), p_1 > 0$. $\beta_1(\cdot)$ and $\alpha_2(\cdot)$ are extended class $\mathcal{K}$ (e.g., $\beta_1(\psi_0(\bm x)) = \psi_0^3(\bm x)$) and class $\mathcal{K}$ (e.g., $\alpha_2(\psi_1(\bm x)) = \psi_1^2(\bm x)$) functions, respectively.
 Since $\dot b(\bm x(0)) > 0$ and $ b(\bm x(0)) < 0$, we can always choose a small enough $p_1$ such that $\psi_1(\bm x(0)) \geq 0$ in (\ref{eqn:example}). The HOCBF constraint (\ref{eqn:constraint}) is the Lie derivative form of $\psi_2(\bm x)\geq 0$ in this case. It follows from Thm. \ref{thm:hocbf} that $\psi_1(\bm x(t))\geq 0, \forall t\geq 0$ if a controller satisfies the corresponding HOCBF constraint (\ref{eqn:constraint}). Because $\beta_1(\cdot)$ is an extended class $\mathcal{K}$ function in (\ref{eqn:example}), the robot will be asymptotically stabilized to the set $C_1:=\{\bm x:b(\bm x)\geq 0\}$, but it will never reach the set boundary in finite time, i.e., the STL specification $\varphi_1$ cannot be satisfied. If both $b(\bm x(0)) < 0$ and $\psi_1(\bm x(0)) < 0$, the HOCBF fails to work since $\psi_1(\bm x)\geq 0$ is not guaranteed to be satisfied in finite time. Since $\psi_1(\bm x)\geq 0$ is equivalent to $\dot \psi_0(\bm x)+ \alpha_1(\psi_0(\bm x))\geq0$ by (\ref{eqn:functions}), we have that the original constraint $b(\bm x)\geq 0$ is also not guaranteed to be satisfied.
We explore how to solve this problem in the next section.

\subsection{High Order Control Lyapunov-Barrier Function}

We introduce HOCLBFs that stabilize a system to a set\footnote{For simplicity, throughout the paper, we say that a system is {\em stabilized to a set} if, when initialized outside the set, it reaches the set in finite time and then it stays inside the set for all future times.}  
defined by $b(\bm x)\geq 0$ whose relative degree is $m$ w. r. t. system (\ref{eqn:affine}). Similar to (\ref{eqn:functions}), we define a sequence of functions:
\begin{equation} \label{eqn:pfunc}
\psi_i(\bm x) := \dot \psi_{i-1}(\bm x) + p_i\beta_i(\psi_{i-1}(\bm x)),\;\;i\in\{1,\dots,m\},
\end{equation}
where $\psi_0(\bm x) := b(\bm x)$ and $p_i \geq 0$. $\beta_i(\cdot),  i\in\{1,\dots,m\}$ are extended class $\mathcal{K}$ functions. 

We also define a sequence of sets as in (\ref{eqn:sets}).  Note that $\bm x(0)\in C_1$ means that system (\ref{eqn:affine}) is initially in the set defined by the constraint $b(\bm x)\geq 0$. If $b(\bm x(0))> 0$, we can always construct a non-empty set $C_1 \cap,\dots, \cap C_{m}$ at time 0 by choosing proper class $\mathcal{K}$ functions in the definition of a HOCBF. Otherwise, there are only some extreme cases (such as $b(\bm x(0))=0$ and $\dot b(\bm x(0)) > 0$) in which we can construct a non-empty set $C_1 \cap,\dots, \cap C_{m}$, as discussed in \cite{Xiao2019}. If we cannot construct such a non-empty set at time 0, we construct $C_1$ as in (\ref{eqn:sets}), and construct sets $C_i, i\in\{2,\dots,m\}$ by (\ref{eqn:pfunc}) and (\ref{eqn:sets}) such that $\bm x(0) \notin C_1 \cap,\dots, \cap C_{m}$. Then we define a HOCLBF as follows:

\begin{definition} \label{def:HOCLBF}
	({\it High Order Control Lyapunov-barrier Function (HOCLBF)}) Let $C_1,\dots, C_{m}$ be defined by (\ref{eqn:sets}) and $\psi_1(\bm x), \dots, \psi_{m}(\bm x)$ be defined by (\ref{eqn:pfunc}). A function $b: \mathbb{R}^n\rightarrow \mathbb{R}$ is a HOCLBF of relative degree $m$ for system (\ref{eqn:affine}) if there exist $(m-i)^{th}$ order differentiable extended class $\mathcal{K}$ functions $\beta_i,i\in\{1,\dots,m-1\}$  and an extended class $\mathcal{K}$ function $\beta_{m}$ such that
	{\small\begin{equation}\label{eqn:cblfcons}
	\begin{aligned}
	\sup_{\bm u\in U}[L_f^{m}b(\bm x) \!+\! L_gL_f^{m-1}b(\bm x)\bm u \!+\! R(b(\bm x)) \!+\! p_m\beta_m(\psi_{m-1}(\bm x))] \geq 0,
	\end{aligned}
	\end{equation}
	}for all $\bm x\in \mathbb{R}^n$. In (\ref{eqn:cblfcons}), $R(\cdot)$ denotes the remaining Lie derivatives along $f$ with degree $< m$ (omitted for simplicity).
\end{definition}

We make the following assumption, which is not true in some cases (such as asymptotically growing functions). However, we will relax it in the next subsection.

\begin{assump} \label{asp:fake}
	If $\psi_{i-1}(\bm x(t)), i\in\{1,\dots,m\}$ is negative at time 0 and there exists a controller $\bm u(t)\in U$ that  makes it strictly increasing $\forall t\geq 0$, then, assume under this controller, $\psi_{i-1}(\bm x(t))$ will become non-negative in finite time.
\end{assump}

\begin{theorem} \label{thm:HOCLBF}
	Given a HOCLBF $b(\bm x)$ from Def. \ref{def:HOCLBF} with the associated sets $C_1, \dots, C_{m}$ defined by (\ref{eqn:sets}), if $\bm x(0) \in C_1 \cap,\dots,\cap C_{m}$, then any Lipschitz continuous controller $\bm u(t)$ that satisfies (\ref{eqn:cblfcons}), $\forall t\geq 0$ renders
	$C_1\cap,\dots, \cap C_{m}$ forward invariant for system (\ref{eqn:affine}). Otherwise, any Lipschitz continuous controller $\bm u(t)$ that satisfies (\ref{eqn:cblfcons}), $\forall t\geq 0$ stabilizes system (\ref{eqn:affine}) to the set $C_1\cap,\dots, \cap C_{m}$.
\end{theorem}

\begin{proof}:
{If  $\bm x\in C_1 \cap,\dots, \cap C_{m}$ for a HOCLBF $b(\bm x)$, then $b(\bm x)$ is also a HOCBF according to Def. \ref{def:hocbf}. It follows from Thm. \ref{thm:hocbf} that the set $C_1\cap,\dots, \cap C_{m}$ is forward invariant for system (\ref{eqn:affine}).} Otherwise, we can define a CLF $V = -\psi_{i-1}(\bm x(t))$ if $\psi_{i-1}(\bm x(t)) < 0, i\in\{1,\dots, m\}$ at time $t$. By the CLF property from \cite{Aaron2012}, if there exists a controller $\bm u(t)\in U$ that satisfies $\psi_{i}(\bm x(t))\geq 0, \forall t$, then system (\ref{eqn:affine}) will be stabilized to the set $C_i$ in (\ref{eqn:sets}). By Asumption \ref{asp:fake}, $\psi_{i-1}(\bm x(t))$ will become non-negative in finite time. This process is done recursively, and any Lipschitz continuous $\bm u(t)$ that satisfies (\ref{eqn:cblfcons}), $\forall t\geq 0$ stabilizes system (\ref{eqn:affine}) to the set $C_1\cap,\dots, \cap C_{m}$. \end{proof}

\subsection{Two Classes of HOCLBFs}
\label{sec:2bf}
In this subsection, we classify HOCLBFs into two classes: one that can achieve finite-time convergence (to a set defined by an arbitrary-relative-degree constraint) if a system is initially outside the set, which can help us relax Assumption \ref{asp:fake} in Thm. \ref{thm:HOCLBF}, and another one that enforces set forward invariance if a system is initially inside the set.

Since power functions are often used for class $\mathcal{K}$ functions, we consider extended class $\mathcal{K}$ functions as power functions. If $q_i = k$ or $q_i = \frac{1}{k}$, where $k\geq 1$ is an odd number, we rewrite (\ref{eqn:pfunc}) in the form:
\begin{equation}\label{eqn:exfunc}
    \psi_i(\bm x) = \dot \psi_{i-1}(\bm x) + p_i \psi_{i-1}^{q_i}(\bm x),
\end{equation}
where $p_i > 0, i\in\{1,\dots,m\}$. Otherwise, we have
\begin{equation}\label{eqn:exfunc2}
    \psi_i(\bm x) =\dot\psi_{i-1}(\bm x) + p_i sign(\psi_{i-1}(\bm x))|\psi_{i-1}(\bm x)|^{q_i}.
\end{equation}
 The sign and absolute value functions are used in the last equation to prevent $\psi_i(\bm x)$ from being an imaginary number when $\psi_{i-1}(\bm x) < 0$.

We only consider (\ref{eqn:exfunc}) in the section. The analysis for (\ref{eqn:exfunc2}) is similar, and thus is omitted. If $q_i \geq 1$, the next lemma shows the asymptotic convergence property of $\psi_{i-1}(\bm x)$ in a HOCLBF (we assume 0 is the initial time WLOG):

\begin{lemma} \label{lem:singleg1}
Given a HOCLBF $b(\bm x)$, if a controller $\bm u(t)\in U$ for (\ref{eqn:affine}) satisfies
\begin{equation} \label{eqn:ineqg1}
\dot \psi_{i-1}(\bm x(t)) + p_i \psi_{i-1}^{q_i}(\bm x(t)) \geq 0, \forall t\geq 0,
\end{equation} 
with $p_i > 0, q_i \geq 1, i\in\{1,\dots,m\}$ and $\psi_{i-1}(\bm x(0)) = \psi_{i-1}^0 \ne 0$, then there exists a lower bound for $\psi_{i-1}(\bm x(t))$, and the lower bound asymptotically approaches 0 as $t\rightarrow \infty$.
\end{lemma}

\begin{proof}
By solving
\begin{equation}\label{eqn:diff}
\dot \psi_{i-1}(\bm x) + p_i \psi_{i-1}^{q_i}(\bm x) = 0,
\end{equation}
with $\psi_{i-1}(\bm x(0))\! =\! \psi_{i-1}^0 \!\ne\! 0$ and $q_i\! =\! 1$, we get
\begin{equation} \label{eqn:exp}
\psi_{i-1}(\bm x(t)) = \psi_{i-1}^0e^{-p_it}.
\end{equation}
In (\ref{eqn:exp}), $\psi_{i-1}(\bm x(t))$ will asymptotically approach 0 as $t\rightarrow +\infty$ for all $\psi_{i-1}^0 \ne 0$. However, $\psi_{i-1}(\bm x(t))$ is always negative if $\psi_{i-1}^0 < 0$, and is always positive if $\psi_{i-1}^0 > 0$.

If $q_i > 1$, the solution of (\ref{eqn:diff}) is given by
\begin{equation} \label{eqn:qne1}
\psi_{i-1}^{1-q_i}(\bm x(t)) = (\psi_{i-1}^0)^{1-q_i} - p_i(1-q_i)t.
\end{equation}
Note that in (\ref{eqn:qne1}), $(\psi_{i-1}^0)^{1-q_i}$ is always positive since $q_i > 1$ and $q_i$ is an odd number. Thus, the right-hand side of (\ref{eqn:qne1}) is always positive. Depending on the sign of $\psi_{i-1}^0$, we have
\begin{equation} \label{eqn:case2}
\psi_{i-1}(\bm x)\!:=\!\left\{
\begin{array}
[c]{rcl}%
-\frac{1}{((\psi_{i-1}^0)^{1-q_i} - p_i(1-q_i)t)^{\frac{1}{q_i-1}}}, \mbox{if $\psi_{i-1}^0 < 0$,}\\
\frac{1}{((\psi_{i-1}^0)^{1-q_i} - p_i(1-q_i)t)^{\frac{1}{q_i-1}}}, \mbox{if $\psi_{i-1}^0 > 0$}.
\end{array}
\right.
\end{equation}
In (\ref{eqn:case2}), $\psi_{i-1}(\bm x)$ will asymptotically get close to 0 from the positive side (if $\psi_{i-1}^0 > 0$) or from the negative side (if $\psi_{i-1}^0 < 0$) as $t\rightarrow \infty$.
Considering (\ref{eqn:ineqg1}) and using the comparison lemma in \cite{Khalil2002}, if $q_i = 1$, we have
\begin{equation} \label{eqn:case10}
\psi_{i-1}(\bm x)\geq\psi_{i-1}^0e^{-p_it}.
\end{equation}
 Otherwise,
\begin{equation} \label{eqn:case22}
\psi_{i-1}(\bm x)\!\geq\!\left\{
\begin{array}
[c]{rcl}%
-\frac{1}{((\psi_{i-1}^0)^{1-q_i} - p_i(1-q_i)t)^{\frac{1}{q_i-1}}}, \mbox{if $\psi_{i-1}^0 < 0$,}\\
\frac{1}{((\psi_{i-1}^0)^{1-q_i} - p_i(1-q_i)t)^{\frac{1}{q_i-1}}}, \mbox{if $\psi_{i-1}^0 > 0$}.
\end{array}
\right.
\end{equation}

It follows from (\ref{eqn:case10}) and (\ref{eqn:case22}) that if $q_i \geq 1$ and $\psi_{i-1}^0 > 0$ ($\psi_{i-1}^0 < 0$), the lower bound of $\psi_{i-1}(\bm x)$ is always positive (negative) and asymptotically approaches 0 as $t\rightarrow \infty$.
\end{proof}

Note that the extended class $\mathcal{K}$ function $p_i\psi_{i-1}^{q_i}(\bm x)$ in (\ref{eqn:exfunc}) is not Lipschitz continuous when $\psi_{i-1}(\bm x) = 0$ if $0< q_i < 1$. Then, we have the following lemma that demonstrates the finite-time convergence property of $\psi_{i-1}(\bm x)$ in a HOCLBF:
\begin{lemma} \label{lem:single}
Given a HOCLBF $b(\bm x)$, if a controller $\bm u(t)\in U$ for (\ref{eqn:affine}) satisfies (\ref{eqn:ineqg1})
with $p_i > 0, q_i \in (0,1), i\in\{1,\dots,m\}$ and $\psi_{i-1}(\bm x(0)) = \psi_{i-1}^0 \ne 0$, then there exists a lower bound for $\psi_{i-1}(\bm x)$, and the time at which this lower bound becomes 0 is $\frac{(\psi_{i-1}^0)^{1-q_i}}{p_i(1-q_i)}$.
\end{lemma}
\begin{proof}
Following (\ref{eqn:qne1}) and $0<q_i < 1$, we have
\begin{equation}\label{eqn:case1}
\psi_{i-1}(\bm x)\!:=\!\left\{
\begin{array}
[c]{rcl}%
\!-\!((\psi_{i-1}^0)^{1-q_i} \!-\! p_i(1\!-\!q_i)t)^{\frac{1}{1-q_i}}, \mbox{if $\psi_{i-1}^0 < 0$,}\\
((\psi_{i-1}^0)^{1-q_i} \!-\! p_i(1\!-\!q_i)t)^{\frac{1}{1-q_i}}, \mbox{if $\psi_{i-1}^0 > 0$}.
\end{array}
\right.
\end{equation}

In (\ref{eqn:case1}), the function $(\psi_{i-1}^0)^{1-q_i} - p_i(1-q_i)t$ will reach 0 at time $t = \frac{(\psi_{i-1}^0)^{1-q_i}}{p_i(1-q_i)}$ and becomes negative after this time instant. Therefore, the values of $\psi_{i-1}(\bm x)$ will be imaginary numbers after $t = \frac{(\psi_{i-1}^0)^{1-q_i}}{p_i(1-q_i)}$.

Using the comparison lemma in \cite{Khalil2002} and considering (\ref{eqn:ineqg1}), since $0<q_i < 1$, we have
\begin{equation}\label{eqn:case11}
\psi_{i-1}(\bm x)\!\geq\!\left\{
\begin{array}
[c]{rcl}%
\!-\!((\psi_{i-1}^0)^{1-q_i} \!-\! p_i(1\!-\!q_i)t)^{\frac{1}{1-q_i}}, \mbox{if $\psi_{i-1}^0 < 0$,}\\
((\psi_{i-1}^0)^{1-q_i} \!-\! p_i(1\!-\!q_i)t)^{\frac{1}{1-q_i}}, \mbox{if $\psi_{i-1}^0 > 0$}.
\end{array}
\right.
\end{equation}

Thus, the lower bound of $\psi_{i-1}(\bm x)$ will be zero at the time instant $t = \frac{(\psi_{i-1}^0)^{1-q_i}}{p_i(1-q_i)}$.
\end{proof}

Motivated by the properties from Lems. \ref{lem:singleg1} and \ref{lem:single}, we classify HOCLBFs into two classes: 
\begin{itemize}	
\item {\it Class 1}: if $\exists i\in\{1,\dots,m\}$, s. t. $0<q_i < 1$ in (\ref{eqn:exfunc}), 
\item {\it Class 2}: $q_i \geq 1, \forall i\in\{1,\dots,m\}$ in (\ref{eqn:exfunc}).
\end{itemize}

\begin{remark}
	({\bf Relationship between CLFs and Class 2 HOCLBFs}): We consider a set $C:=\{\bm x: b(\bm x)\geq 0\}$ where $b(\bm x)$ is a Class 2 HOCLBF with $m=1$. The HOCLBF constraint (\ref{eqn:cblfcons}) in this case is given by
	$
	\dot b(\bm x) + p_1b^{q_1}(\bm x)\geq 0, q_1 \geq 1, p_1 \geq 0.
	$
	Assuming $b(\bm x(0)) < 0$, we shrink the set $C$ to a point, i.e., $C:=\{\bm x: b(\bm x)= 0\}$.  Let $V(\bm x) := -b(\bm x) > 0$. We can rewrite the last equation as:
	$
	\dot V(\bm x) + p_1V^{q_1}(\bm x)\leq 0,
	$
	where the positive definite property of $V(\bm x)$ is guaranteed by the Class 2 HOCLBF as shown in (\ref{eqn:case10}) and (\ref{eqn:case22}) (note that $q_1$ is an odd number). This equation is equivalent to the CLF constraint defined in \cite{Aaron2012}.  Therefore, the  {\it Class 2} HOCLBF is more general than a CLF, and it works for high-relative-degree systems.
\end{remark}

Next, we continue to consider the {\it Class 1} HOCLBF to show its finite-time convergence property with the above lemmas. If $\psi_j(\bm x(t_i))\geq 0, \forall j\in\{i,\dots,m\}$, where $i\in\{1,\dots,m\}, t_i\geq 0$, then we can define $\psi_i(\bm x)$ as a HOCBF to guarantee that $\psi_j(\bm x(t))\geq 0, \forall j\in\{i,\dots,m\}, \forall t\geq t_i$ \cite{Xiao2019}. Thus, we assume that  a {\it Class 1} HOCLBF always defines $\psi_i(\bm x)$ to be a HOCBF if $\psi_j(\bm x(t_i))\geq 0, \forall j\in\{i,\dots,m\}$ as it better guarantees finite-time convergence.

Given a {\it Class 1} HOCLBF $b(\bm x)$ with $b(\bm x(0)) < 0$ and $\psi_{i}(\bm x(0)) = \psi_{i}^0 \in\mathbb{R}, i\in\{1,\dots,m-1\}$, we define
\begin{equation} \label{eqn:m0}
    m_0\!=\! \left\{\!\!\!\!\!\!\!\!\!
\begin{array}
[c]{rcl}%
&\mathop{\min}\limits_{i\in\{1,\ldots,m-1\}:\psi_{i}^0>0} i, \quad \mbox{if there exists $i$ s.t. $\psi_{i}^0>0$}\\
&m,\qquad\qquad\quad\;\; \mbox{otherwise}.
\end{array}
\right.
\end{equation}
In summary, if $i\leq m_0$, we choose $q_i\in(0,1)$ in (\ref{eqn:exfunc}); otherwise, we choose $q_i\geq 1$ for a {\it Class 1} HOCLBF.

 Let $t_i\geq0, i\in\{1,\dots,m\}$ denote the starting time instant when $\psi_j(\bm x(t_i))\geq 0, \forall j\in\{i,\dots,m\}$. Each $t_i$ depends on $\bm x(0)$ and $\bm u(t), t\geq 0$. The following theorem provides the finite-time convergence property of a {\it Class 1} HOCLBF:

\begin{theorem} \label{thm:time}
	Given a {\it Class 1} HOCLBF $b(\bm x)$ with $b(\bm x(0)) < 0$, any controller $\bm u(t)\in U$ that satisfies (\ref{eqn:cblfcons}) makes (\ref{eqn:affine}) converge to the set $C_1\cap\dots\cap C_m$ within time
	\begin{equation}
	t_{up} = \sum_{i=1}^{m_0}\frac{(\psi_{i-1}(\bm x(t_i)))^{1-q_i}}{p_i(1-q_i)}.
	\end{equation}
\end{theorem}
\begin{proof}
	If $m_0 = m$, then we have $p_i>0, q_i\in(0,1),\forall i\in\{1,\dots,m\}$, and a controller that satisfies the HOCLBF constraint (\ref{eqn:cblfcons}) will drive $\psi_{i-1}(\bm x)\geq 0$ from $i = m$ to $i=1$ following from Lem. \ref{lem:single}. Thus, we have that the time for $b(\bm x)=0$ is bounded by $\sum_{i=1}^{m}\frac{(\psi_{i-1}(\bm x(t_i)))^{1-q_i}}{p_i(1-q_i)}$. Otherwise, since $\psi_{m_0}^0 > 0$, by choosing proper $p_j>0, q_j\geq 1, j\in\{m_0,\dots,m\}$, we can get a non-empty set $C_{m_0+1}\cap,\dots,\cap C_m$ \cite{Xiao2019}. Then $\psi_{m_0}(\bm x)\geq 0$ is guaranteed by Thm. \ref{thm:hocbf}. By Lem. \ref{lem:single}, the upper bound time for each $\psi_{i-1}(\bm x) =0, i\in\{1,\dots,m_0\}$ in (\ref{eqn:exfunc}) is given by $\frac{(\psi_{i-1}(\bm x(t_i)))^{1-q_i}}{p_i(1-q_i)}$.  In this case, we need to recursively drive $\psi_{i-1}(\bm x)$ to 0 within time $\frac{(\psi_{i-1}(\bm x(t_i)))^{1-q_i}}{p_i(1-q_i)}$ from $i = m_0$ to $i=1$. Thus, the time for system (\ref{eqn:affine}) to converge to $C_1\cap\dots\cap C_m$ is upper bounded by $\sum_{i=1}^{m_0}\frac{(\psi_{i-1}(\bm x(t_i)))^{1-q_i}}{p_i(1-q_i)}$.
	\end{proof}

\begin{remark}
(Chattering in {\it Class 1} HOCLBFs) By Lemma \ref{lem:single}, we have that $\psi_{i-1}(\bm x)$ will go to zero within time $t = \frac{(\psi_{i-1}(\bm x(t_i)))^{1-q_i}}{p_i(1-q_i)}$ when $\psi_{i-1}(\bm x(t_i))$ is negative. This could also be true when (\ref{eqn:ineqg1}) becomes active if $\psi_{i-1}(\bm x(t_i))$ is positive, which is usually imposed by the state convergence requirement. After $\psi_{i-1}(\bm x)$ becomes zero, it will become positive (negative) if it is initially negative (positive) due to the continuity of the dynamics (\ref{eqn:affine}). However, $\psi_{i-1}(\bm x)$ may go to zero again after it becomes positive (negative) following from (\ref{eqn:case11}) as discussed earlier. Recursively, this may cause a {\it chattering} behavior.
\end{remark}

We can relax Assumption \ref{asp:fake} by defining a {\it Class 1} HOCLBF when $\bm x(0)\notin C_1\cap\dots\cap C_m$ since $\psi_{i-1}(x(t))$ will always cross the boundary $\psi_{i-1}(\bm x(t)) = 0$ in finite time when $\dot\psi_{i-1}(\bm x)> 0$, (a condition imposed by $\psi_{i}(\bm x(t)) \geq 0$ in Def. \ref{def:HOCLBF}).  After $\psi_{i-1}(\bm x)$ becomes positive, we can re-define an extended power class $\mathcal{K}$ function with $q_i \geq 1$ for $\psi_{i}(\bm x)$ in (\ref{eqn:exfunc}) in order to eliminate the chattering behavior. This switch process is discussed in the following remark.

\begin{remark} \label{rem:switch}
({\bf Switch from Class 1  to Class 2 HOCLBF}) A {\it Class 1} HOCLBF is defined when $b(\bm x(0)) \leq 0$, and the switch from a Class 1 to a Class 2 HOCLBF is performed as follows. Recall that $\psi_0(\bm x) = b(\bm x)$ in (\ref{eqn:pfunc}). If $\psi_{i-1}^0\leq 0, \forall i\in\{2,\dots,m\}$, we choose $p_i> 0, q_i\in(0,1)$ from $i = 1$ to $i = m$, and have the HOCLBF constraint (\ref{eqn:cblfcons}). $\psi_{i-1}(\bm x)$ will be non-negative from $i = m$ to $i = 1$ in finite time by Thm. \ref{thm:time}. Otherwise, we also choose $p_i> 0, q_i\in(0,1)$ starting from $i = 1$. Suppose $\psi_{i-1}^0 > 0$ at some $i\in\{2,\dots,m\}$. Then we can always choose $p_j > 0, q_j\geq 1,\forall j\in\{i,\dots,m\}$ such that  $C_{i}\cap,\dots,\cap C_m$ is non-empty \cite{Xiao2019}. It follows from Thm. \ref{thm:HOCLBF} that $\psi_{i-1}(\bm x) \geq 0$ is guaranteed if the HOCLBF constraint  (\ref{eqn:cblfcons}) is satisfied. Thus, $\psi_{i-2}(\bm x)$ will be positive in finite time following from Lem. \ref{lem:single} and the continuity of system (\ref{eqn:affine}). Once $\psi_{i-2}(\bm x)$ becomes positive, we change $q_{i-1}\in(0,1)$ to $q_{i-1}\geq 1$ and choose $p_{i-1}$ such that $C_{i-1}\cap, \dots,\cap C_{m}$ is non-empty. This is done recursively until $b(\bm x)> 0$. Eventually, we have $q_i\geq 1, \forall i\in \{1,\dots,m\}$. The {\it Class 1} HOCLBF switches to a {\it Class 2} HOCLBF.
\end{remark}

In a nutshell, we would like to define a {\it Class 1} HOCLBF when $b(\bm x(0)) \leq 0$ as the state of system (\ref{eqn:affine}) will converge to the set $C_1 \cap,\dots, \cap C_{m}$ without Assumption \ref{asp:fake} in finite time, and define a {\it Class 2} HOCLBF when $b(\bm x(0)) > 0$ in which case we can always define $C_i, i \in\{1,\dots,m\}$ such that $\bm x(0)\in C_1 \cap,\dots, \cap C_{m}$, as shown in \cite{Xiao2019}. Then the set $C_1 \cap,\dots, \cap C_{m}$ is forward invariant, as shown in Thm. \ref{thm:HOCLBF}. If we want $\psi_{i-1}(\bm x)$ to decrease to 0 slower, we can define a {\it Class 2} HOCLBF with large $q_i$ value, as shown in (\ref{eqn:case22}).

\subsection{HOCLBFs for STL Satisfaction}
\label{sec:stl_holbf}
In this section, we show how we can use HOCLBFs to guarantee the satisfaction of a STL formula. A STL formula can be decomposed into atomic formulae composed of $\mathcal{G}, \mathcal{F}$ operators, and each atomic formula is mapped to a constraint over the state of (\ref{eqn:affine}). Starting from time 0, we formulate a receding horizon $H > 0$, and only consider the atomic formulae that are in this horizon, i.e., $[t,t+H]\cap[t_a,t_b]\ne \emptyset$. If the constraint is satisfied at the current state, we can define a {\it Class 2} HOCLBF to make sure the predicate always stays true. The implementation is the same as for HOCBF, and thus is omitted. If this constraint is violated at the current state, we can use {\it Class 1} HOCLBFs to guarantee it to be satisfied within specified time. Once this constraint is satisfied, we switch to a {\it Class 2} HOCLBF as shown next.

{\bf Always atomic formula $\mathcal{G}$}: $\bm x\models\varphi$, where $\varphi:= \mathcal{G}_{[t_a,t_b]}(||\bm x(t) - \bm K|| \leq \xi)$, $\bm K\in\mathbb{R}^n, 0\leq t_a\leq t_b$, and $\xi > 0$, requires the trajectory $\bm x$ of system (\ref{eqn:affine}) to satisfy the quantified constraint:
\begin{equation} \label{eqn:satisf1}
\forall t\in[t_a,t_b], \;\;\;||\bm x(t) - \bm K|| \leq \xi.
\end{equation}

 Let $b(\bm x):= \xi - ||\bm x -\bm K||$, where $b(\bm x)$ has relative degree $m$ for system (\ref{eqn:affine}) and $b(\bm x(0)) < 0$. If we define $b(\bm x)$ to be a {\it Class 1} HOCLBF and choose $p_i > 0, q_i\in(0,1), i\in\{1,2,\dots, m_0\}$ to satisfy 
\begin{equation}\label{eqn:Geq}
t_{a} \geq \sum_{i=1}^{m_0}\frac{(\psi_{i-1}(\bm x(t_i)))^{1-q_i}}{p_i(1-q_i)},
\end{equation}
then the constraint (\ref{eqn:satisf1}) is guaranteed to be satisfied at $t_a$ following from Thm. \ref{thm:time} and is always satisfied after $t_a$ when we define $b(\bm x)$ to be a {\it Class 2} HOCLBF when (\ref{eqn:satisf1}) is satisfied to avoid chattering. We remove the HOCLBF $b(\bm x)$ after $t_b$. Thus, this atomic formula is guaranteed to be satisfied.
Since $\psi_{i-1}(\bm x(t_i)), i\in\{2,\dots,m\}$ depends on $p_j, q_j, \forall j\in[1,\dots,i]$, choosing $p_i, q_i$ to satisfy constraint (\ref{eqn:Geq}) is difficult. However, this can be easily solved if we define an Adaptive CBF (AdaCBF) \cite{Xiao2020} that makes $p_i, q_i$ time-varying (adaptive). In this paper, we provide a simple approach to choose $p_i, q_i$, i.e., we redefine $\psi_{i}(\bm x)$ in (\ref{eqn:exfunc}) as ($p_i > 0$):
\begin{equation}
\begin{aligned} \psi_i(\bm x) :=& \left\{ \begin{array}{lll} \dot\psi_{i-1}, \text{ if }i<m_0, \\ \dot\psi_{i-1}(\bm x) + p_i \psi_{i-1}^{q_i}, q_i\in(0,1), \text{ if }i=m_0, \\ \dot\psi_{i-1}(\bm x) + p_i \psi_{i-1}^{q_i}, q_i\geq 1, \text{ otherwise}. \end{array} \right.\\ \end{aligned} \label{eqn:gfunctions}%
\end{equation}

Now, $\psi_{m_0}(\bm x)$ in (\ref{eqn:gfunctions}) excludes $p_i, q_i, \forall i\in\{1,\dots,m_0 - 1\}$. We partition the time $[0,t_a]$ into $m_0$ intervals $\{t_1,\dots, t_{m_0}\}$ such that $\sum_{i = 1}^m t_i = t_a$. Each interval corresponds to the time necessary to drive $\psi_{i-1}(\bm x), i\in\{1,\dots,m_0\}$ in (\ref{eqn:gfunctions}) from negative to positive. We update $m_0\leftarrow m_0-1$ whenever $\psi_{m_0 - 1}(\bm x)> 0$, and then design each pair of $p_{m_0}, q_{m_0}$ according to Lem. \ref{lem:single} and the pre-partitioned time interval mentioned above. Each pair of $p_i, q_i$ is determined online, as summarized in Algo. \ref{alg:always}.
\begin{algorithm}
	\caption{The satisfaction of a $\mathcal{G}$ atomic formula} \label{alg:always}
	\KwIn{ $\varphi$, system (\ref{eqn:affine}) with initial state $\bm x(0)$}
	\KwOut{$\bm p,\bm q$}
	Determine $m_0$ by (\ref{eqn:m0})\;
	Partition $[0,t_a]$ into $m_0$ intervals $\{t_1,\dots, t_{m_0}\}$\;
	
	\While{$t\leq T$}{
		Determine $p_{m_0}, q_{m_0}$ according to Lem. \ref{lem:single} and each corresponding interval $t_{m_0}$\;
		Determine $p_i, q_i, i\in\{m_0+1,\dots,m\}$ by (\ref{eqn:gfunctions})\;
		Formulate a CBF-based OCP as in Prob. \ref{prob:tocp}\;
		\While{ $t\leq T$}
		{Solve the OCP in Prob. \ref{prob:tocp}\;
    		\If{$\psi_{m_0 - 1}(\bm x) > 0$}
    		{
    		$m_0\leftarrow m_0 - 1$ and break\;
    		}
		}
	}
\end{algorithm}

\textbf{Example revisited.} For the robot control problem in Sec. \ref{sec:pre}, consider formula $\varphi_1$, which corresponds to the constraint:
\begin{equation} \label{eqn:robotstl1}
\forall t\in[5s,6s], \;\;\;x^2(t) + y^2(t) \leq R^2. 
\end{equation}

Let $b(\bm x) = R^2 - x^2 - y^2$ be a {\it Class 1} HOCLBF. The initial condition of system (\ref{eqn:robot}) is given by $(0, -7.7,\frac{\pi}{4})$, $R = 4m, v = 1.732m/s$. We have $b(\bm x(0)) = -43.29$ and $\dot b(\bm x(0)) > 0$, and thus, $m_0 = 1$. If we choose $p_1 = 5, q_1 = \frac{1}{3}, t_1 = 4s$, then $\psi_1(\bm x(0)) = 1.3042 > 0$, and $t_1 > \frac{(b(\bm x(0)))^{1-q_1}}{p_1(1-q_1)}$ is satisfied. Thus, by Thm. \ref{thm:time}, the formula $\varphi_1$ is guaranteed to be satisfied.

{\bf Eventually atomic formula $\mathcal{F}$}: $\bm x\models\varphi$, where $\varphi:= \mathcal{F}_{[t_a,t_b]}(||\bm x(t) - \bm K|| \leq \xi)$, $\bm K\in\mathbb{R}^n, 0\leq t_a\leq t_b$, and $\xi > 0$, requires the trajectory $\bm x$ of system (\ref{eqn:affine}) to satisfy the quantified constraint:
\begin{equation} \label{eqn:satisf2}
\exists t\in[t_a,t_b], \;\;\;||\bm x(t) - \bm K|| \leq \xi.
\end{equation}

 Let $b(\bm x):= \xi - ||\bm x -\bm K||$, where $b(\bm x)$ has relative degree $m$ for system (\ref{eqn:affine}) and $b(\bm x(0)) < 0$. 
If we define $b(\bm x)$ to be a {\it Class 1} HOCLBF and choose $p_i > 0, q_i\in(0,1), i\in\{1,\dots, m_0\}$ to satisfy 
\begin{equation}
t_{b} \geq \sum_{i=1}^{m_0}\frac{(\psi_{i-1}(\bm x(t_i)))^{1-q_i}}{p_i(1-q_i)},
\end{equation}
then constraint (\ref{eqn:satisf2}) is guaranteed to be satisfied before $t_b$ following from Thm. \ref{thm:time}. If the predicate $b(\bm x(t))\geq 0$ is satisfied before $t_a$, then we will switch to a {\it Class 2} HOCLBF to make the predicate stay true. We remove the HOCLBF $b(\bm x)$ once the constraint (\ref{eqn:satisf2}) is satisfied for any time instant in $[t_a,t_b]$. In this way, this atomic formula is guaranteed to be satisfied. The approach to choose $p_i, q_i$ is similar to the $\mathcal{G}$ atomic formula.

{\bf Disjunction, conjunction, and Until formulae:}  For conjunctions of atomic formulae, we consider the corresponding HOCLBFs at the same time. We also consider the corresponding HOCLBFs at the same time for the disjunctions of atomic formulae. However, we will relax the one whose barrier function value is smaller when any two of the atomic formulae conflict and remove all the HOCLBFs once any one of these HOCLBFs is non-negative. Note that an Until formula $\mathcal{U}$ is a conjunction of $\mathcal{G}$ and $\mathcal{F}$ atomic formulae \cite{Maler2004}.

\textbf{Horizon $H$ and conflict predicates:} The horizon $H$ (see the description of the approach in Sec. \ref{sec:prob}) is chosen as large as possible given the available computation resources. While we define a HOCLBF for each atomic formula, it is likely that there will be conflict predicates among the predicates within $H$, which could make the problem infeasible. To address this, we relax the predicates in formulae with larger $t_a$, while minimizing the relaxation (relaxing a predicate $b(\bm x)\geq 0$ means relaxing the corresponding HOCLBF constraint (\ref{eqn:cblfcons}) by replacing 0 in the right-hand side with $\delta \in\mathbb{R}$ and adding $\delta^2$ to the cost).

\textbf{Solution to Problem \ref{prob:pre}:} The {\it Class 1} and {\it Class 2} HOCLBFs guaranteeing the satisfaction of the STL formula $\varphi$ are added as constraints to Problem \ref{prob:tocp}, which is solved using a sequence of QPs as described at the end of Sec. 
\ref{sec:hocbf}.

\section{Case Study}
\label{sec:sim}

Consider the unicycle described by Eqn. (\ref{eqn:robot}). The objective is to minimize the control effort:
$
	\min_{\bm u(t)} \int_{0}^{T} u^2(t) dt
$. The STL specification is given by
\begin{equation}
  \bm x\models (\varphi_1 \Rightarrow \varphi_2)\wedge( \varphi_0  \Rightarrow \varphi_3) \wedge \varphi_4 \wedge \varphi_5 \wedge \varphi_6\wedge\varphi_7,  
\end{equation}
where $\varphi_0 := (b_1(\bm x(0)) < 0), \varphi_1 := (b_1(\bm x(0)) \geq 0)$, $\varphi_2 =  G_{[0,t_b]}(b_1(\bm x)\geq 0)$, $\varphi_3 =  G_{[t_a,t_b]}(b_1(\bm x)\geq 0), \varphi_4 =F_{[t_c,t_d]}b_2(\bm x) \geq 0, \varphi_5 = G_{[t_e,T]}(b_3(\bm x)\geq 0), \varphi_6 = G_{[0,T]}(b_4(\bm x)\geq 0), \varphi_7 = G_{[0,T]}(b_5(\bm x)\geq 0), 0<t_a<t_b<t_c<t_d<t_e<T$, where
\begin{equation}\label{eqn:robotdst}
b_1(\bm x) := R_1^2 - x^2 - y^2 \geq 0,
\end{equation}
\begin{equation}\label{eqn:robotdes1}
b_2(\bm x) := \phi^2-(\theta - \theta_d)^2 \geq 0,
\end{equation}
\begin{equation}\label{eqn:robotdes2}
b_3(\bm x) :=  R_2^2 - (x + A_x)^2 + (y + A_y)^2 \geq 0,
\end{equation}
describe desired sets, with $R_1 > 0, R_2 > 0,\phi>0,\theta_d \in\mathbb{R}, A_x\in\mathbb{R},A_y\in\mathbb{R}$. Functions $b_4(\bm x)$ and $b_5(\bm x)$ describe two obstacles, i.e.,
\begin{equation}\label{eqn:robotdes3}
b_4(\bm x) :=  (x + O_{x,1})^2 + (y + O_{y,1})^2- R_3^2 \geq 0,
\end{equation}
\begin{equation}\label{eqn:robotobs}
b_5(\bm x) := (x + O_{x,2})^2 + (y + O_{y,2})^2- R_4^2 \geq 0,
\end{equation}
where {\small$R_3 > 0, R_4 > 0, (O_{x,1},O_{y,1})\in\mathbb{R}^2, (O_{x,2},O_{y,2})\in\mathbb{R}^2$.}

  In plain English, the STL specification states that, if the robot is initially in the set defined by constraint (\ref{eqn:robotdst}), then it should stay there for all times in the interval $[0,t_b]$. Otherwise, it should stay in this set for all times in $[t_a,t_b]$. The heading of the robot should be $\theta_d$ with error $\phi$ for at least a time instant in $[t_c,t_d]$, and the robot should stay in the set defined by the constraint (\ref{eqn:robotdes2}) for all times in $[t_e,T]$. The robot should always avoid the obstacles defined by (\ref{eqn:robotobs}) (\ref{eqn:robotdes3}).

 The control limitation is defined as:
$
u_{min}\leq u\leq u_{max},
$
where $u_{min}<0, u_{max}>0$.
Since the relative degrees of all the constraints (\ref{eqn:robotdst})-(\ref{eqn:robotobs}) with respect to (\ref{eqn:robot}) are 2, we define HOCLBFs with $m = 2$ to implement the STL specifications. We solve the OCP with the approach introduced at the end of Sec. \ref{sec:gcbf}.

We implemented the proposed algorithms in MATLAB. We used Quadprog to solve the QPs and ODE45 to integrate the robot dynamics. We first present simulations for initially violated constraints to study {\it Class 1} and {\it Class 2} HOCLBFs, and then present the complete solution to the OCP with STL specifications.

\subsection{Finite-time Convergence}

We consider the atomic formula $\varphi_3$ to study both {\it Class 1} and {\it Class 2} HOCLBFs for an initially violated constraint. The robot initial state is given by $(0,-7.7, \frac{\pi}{4}), v = 1.732m/s$, and is initially out of $C_1:=\{\bm x:b_1(\bm x)\geq 0\}$. Other simulation parameters are $t_b = 30s, \Delta t = 0.1, u_{max} = -u_{min} = 0.6rad/s, R = 4m$. We first define one {\it Class 1} HOCLBF and two {\it Class 2} HOCLBFs (linear and quadratic, respectively) for the constraint (\ref{eqn:robotdst}), and study the finite-time convergence under different $p_1, p_2$, respectively. The simulation results are shown in Fig. \ref{fig:stab1}.

\begin{figure}[thpb]
	\centering
	\includegraphics[scale=0.5]{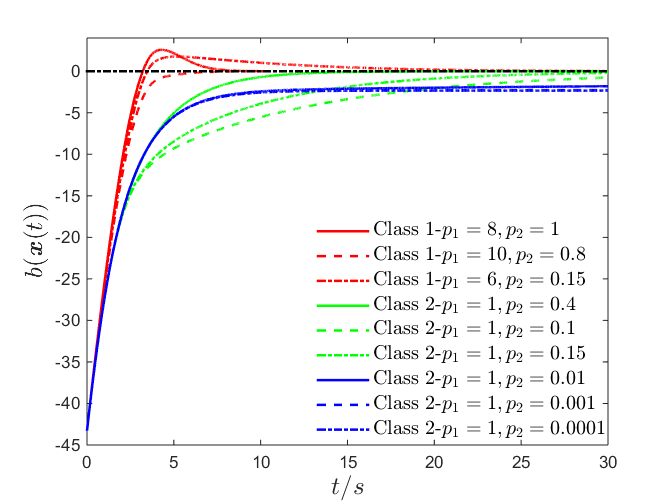}
	\caption{Finite-time convergence (corresponding to $C_1$) of the system under different classes of HOCLBF. All the {\it Class 1} HOCLBFs have powers $q_1 =q_2= \frac{1}{3}$, and the  {\it Class 2} HOCLBFs (green and blue, respectively) have $q_1 =q_2= 1$ and $q_1= 1, q_2= 2$, respectively.}	
	\label{fig:stab1}
\end{figure}

It follows from Fig. \ref{fig:stab1} that the robot can enter $C_1$ with {\it Class 1} HOCLBFs. In {\it Class 2} HOCLBFs, both $b(\bm x)$ and $\psi_1(\bm x)$ will asymptotically approach 0, and remain negative, i.e., the robot can never enter the set $C_1$. The convergence speed depends heavily on the penalties. The robot may even be stabilized to a distance that is far away from $C_1$ under high order power functions, as the blue lines shown in Fig. \ref{fig:stab1}. Note that after $b(\bm x)$ becomes positive for {\it Class 1} HOCLBFs, there will be chattering behaviors that could easily make the QP infeasible.

\subsection{Chattering Behavior}

We consider {\it Class 1} HOCLBFs to study chattering behaviors. The robot starts inside the set $C_1:=\{\bm x:b_1(\bm x)\geq 0\}$ with $\bm x(0) = (0,-3.7, 0), v = 1.732m/s$. The other settings are the same as in the last subsection. There would be chattering for the robot if we define a {\it Class 1} HOCLBF for the safety constraint (\ref{eqn:robotdst}), as the blue curves shown in Fig.  \ref{fig:control2}. In order to avoid chattering, we switch a {\it Class 1} HOCLBF to a {\it Class 2} HOCLBF, as shown in Remark \ref{rem:switch}. For the three {\it Class 1} HOCLBFs in Fig. \ref{fig:stab1}, we show the BF profiles with the switch method to avoid chattering in Fig. \ref{fig:switch}.

\begin{figure}[htbp]
	\centering
\hspace{-6mm}	\subfigure[Chattering behaviors.]{
		\begin{minipage}[t]{0.45\linewidth}
			\centering
			\includegraphics[scale=0.29]{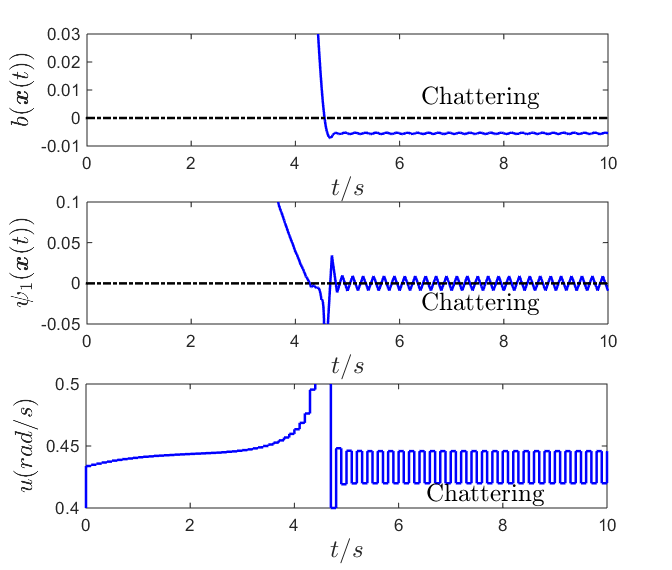} 
			\label{fig:control2}%
		\end{minipage}%
	}	
	\subfigure[The switch method. ]{
		\begin{minipage}[t]{0.45\linewidth}
			\centering
			\includegraphics[scale=0.33]{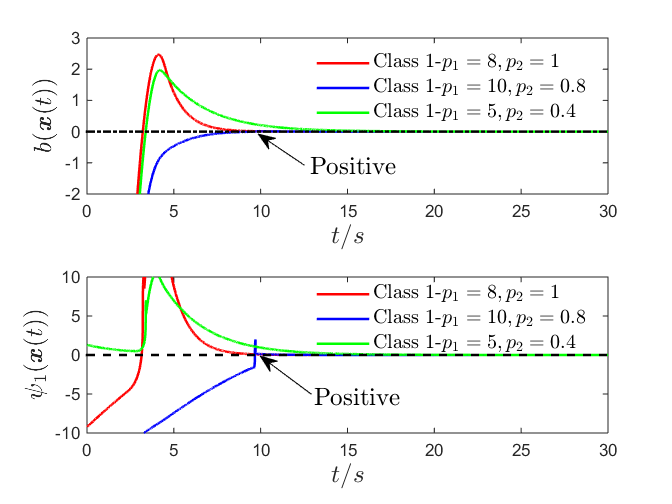} 
			\label{fig:switch}%
		\end{minipage}%
	}	
	
	\centering
	\caption{Chattering behaviors ($p_1 = 6, p_2 = 0.14, q_1 = q_2 = \frac{1}{3}$) and the switch method for {\it Class 1} HOCLBFs.}
	\vspace{-3mm}
\end{figure}

\subsection{Complete Solution}

For each atomic formula, we find the corresponding $p_1,p_2, q_2,q_2$ using the approach introduced in Sec. \ref{sec:stl_holbf}. The simulation parameters are $T = 32s, t_a = 4s, t_b = 5s, t_c = 7s, t_d = 9s, t_e = 21s, \Delta t= 0.1s, R_1 = 4m, R_2 = 4m, R_3 = 2m, R_4 = 3m, Ax = 10m, A_y = 10m, \phi = \frac{\pi}{12}, \theta_d =  \frac{5\pi}{4}, O_{x,1} = 8m, O_{y,1} = 4m, O_{x,2} = 10m, O_{y,2} = 10m, u_{\max} = -u_{\min} = 0.9rad/s, v = 1.732m/s, H = 10s$. The robot initial state is $(0,-7.7,\frac{\pi}{4})$.
\begin{figure}[thpb]
	\centering
	\vspace{-4mm}
	\includegraphics[scale=0.5]{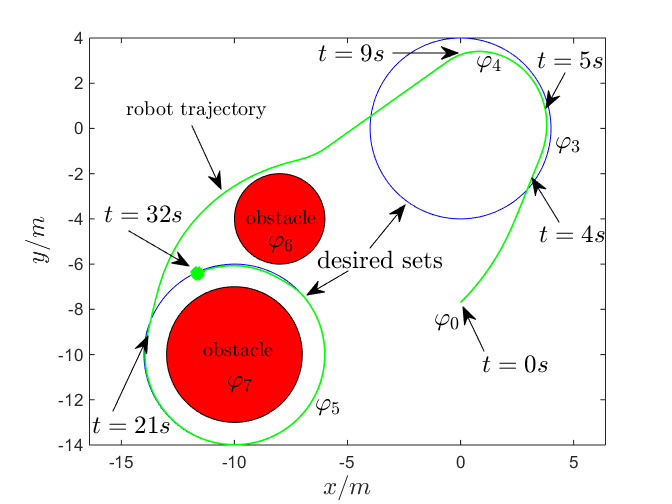}
	\vspace{-4mm}
	\caption{A trajectory that satisfies the STL specification with HOCLBFs.}
	\label{fig:traj}
\end{figure}

We choose $q_1 =q_2 = \frac{1}{3}$ for all {\it Class 1} HOCLBFs, and choose $q_1 = q_2 = 1$ for all {\it Class 2} HOCLBFs.  Then we get $(p_1, p_2)$ with the approach introduced in Sec. \ref{sec:stl_holbf} as $(5,0.4), (0.8,N/A), (4.85,0.4)$ for the atomic formulae $\varphi_3, \varphi_4, \varphi_5$, respectively. Note that the relative degree of (\ref{eqn:robotdes1}) is one, so $\varphi_4$ only has $p_1$. The $p_1, p_2$ for $\varphi_6, \varphi_7$ are chosen according to the penalty method \cite{Xiao2019} such that the QP is feasible. When the {\it Class 1} HOCLBF constraint (desired set) conflicts with the {\it Class 2} HOCLBF constraint (safety), we relax the {\it Class 1} HOCLBF constraint. After this conflict disappears, we check whether the current $p_1, p_2$ can still satisfy the atomic formula or not. If not, we need to redefine $p_1,p_2$. The STL specification is guaranteed to be satisfied, as shown in Fig. \ref{fig:traj}.
\section{Conclusion}
\label{sec:conc}
We propose high order control Lyapunov-barrier functions (HOCLBF) that work for constraints with arbitrary relative degree and systems with arbitrary initial state. We show how the proposed HOCLBFs can be used to enforce the satisfaction of Signal Temporal Logic (STL) specifications.  Simulation results on a unicycle model demonstrate the effectiveness of the proposed method. Future work will focus on feasibility under tight control bounds and robust satisfaction of STL specifications.

\bibliographystyle{plain}
\bibliography{MCBF}

\end{document}